\documentclass[letterpaper, 10 pt, conference]{ieeeconf} 
\IEEEoverridecommandlockouts                              
\usepackage[utf8]{inputenc}
\usepackage{mathtools}
\mathtoolsset{showonlyrefs,showmanualtags}
\usepackage{amssymb,amsthm}
\usepackage{dsfont}
\usepackage{booktabs}
\usepackage{color}
\usepackage{bbold} 
\usepackage[yyyymmdd,hhmmss]{datetime}
\usepackage{bm,upgreek}
\usepackage{tablefootnote}
\usepackage{nth}

\usepackage[ruled]{algorithm} 
\usepackage{algpseudocode}
\usepackage{algorithm}

\newtheorem{assumption}{Assumption}
\newtheorem{theorem}{Theorem}

\newtheorem{remark}{Remark}
\newtheorem{lemma}{Lemma}

\newtheorem{problem}{Problem}

\DeclareMathOperator{\DIAG}{diag}
\DeclareMathOperator{\VEC}{vec}
\DeclareMathOperator{\ROW}{row}
\DeclareMathOperator{\TR}{Tr}

\DeclareMathOperator{\Emp}{\mathcal{E}}
\DeclareMathOperator{\Binopdf}{binopdf}
\newcommand{\Ninf}[1]{ \| #1 \| }
\newcommand{\ID}[1]{ \mathbb{1} (#1 ) }

\newcommand{\Exp}[1]{\mathbb{E}\big[ #1\big]}
\newcommand{\Prob}[1]{\mathbb{P} \big(#1\big)}       
\newcommand{\argmin}{\operatornamewithlimits{argmin}}

\title{\LARGE \bf
Reinforcement Learning in Deep Structured Teams: Initial Results with Finite and Infinite Valued Features
}

\author{Jalal Arabneydi, Masoud Roudneshin and Amir G. Aghdam
\thanks{This work has been supported by the Natural Sciences and Engineering Research Council of Canada (NSERC) under Grant RGPIN-262127-17.}  
\thanks{Jalal Arabneydi, Masoud Roudneshin and Amir G. Aghdam are with the  Department of Electrical and Computer Engineering, 
        Concordia University, 1455 de Maisonneuve Blvd. West, Montreal, QC, Canada, Postal Code: H3G 1M8.  Email:  {\tt\small jalal.arabneydi@mail.mcgill.ca},    {\tt\small m\_roudne@encs.concordia.ca},   
        {\tt\small aghdam@ece.concordia.ca}}%
}

\begin{document}

\maketitle
\vspace*{-5cm}{\footnotesize{Proceedings of IEEE  Conference on Control Technology and Applications, 2020. Doi: 10.1109/CCTA41146.2020.9206397.}}
\vspace*{3.85cm}

\thispagestyle{empty}
\pagestyle{empty}
\begin{abstract}
In this paper, we consider Markov chain and linear quadratic  models for deep structured teams with  discounted  and time-average cost functions under two non-classical information structures, namely, deep state sharing and no sharing. In deep structured teams, agents are coupled in dynamics and cost functions through deep state, where deep state  refers to a set of orthogonal linear regressions of the states. In this article, we consider a homogeneous linear regression for Markov chain models (i.e.,  empirical distribution  of states)  and a few orthonormal linear regressions  for linear quadratic models (i.e., weighted average of states).  Some  planning algorithms   are developed for the case when the model is known, and some  reinforcement learning algorithms are proposed for the case when the model is not known completely. The convergence of  two model-free (reinforcement learning) algorithms, one for Markov chain models and one for linear quadratic models,  is  established.   The results   are then applied to a smart grid.
\end{abstract}

\section{Introduction}
Recently, there has been a surge of interest in the application of reinforcement learning algorithms in networked control systems such as social networks,  swarm robotics, smart grids and transportation networks. This type of  systems often consist of many interconnected agents (decision makers) that wish  to perform  a common task  with   limited resources in terms of computation, information and knowledge.  

When every agent has perfect information and complete knowledge of the entire network, the optimal solution can be computed by  dynamic programming decomposition. The computational complexity of  solving this dynamic program increases with the  number of agents (the so-called  ``curse of dimensionality'').  The above complexity is drastically exacerbated when the information is imperfect. For the case of decentralized information structure with finite spaces, on the other hand, the computational complexity of the resultant dynamic program  is NEXP~\cite{Bernstein2002complexity}, and  for infinite spaces with linear quadratic model, the optimization problem is non-convex~\cite{Witsenhausen1968Counterexample}. In addition,   the underlying network  model  is not always known completely; this lack of knowledge further  increases the above complexity.  Subsequently,  it is very difficult to solve a large-scale control problem with imperfect information of agents and  incomplete knowledge of the network. 

As an attempt to address the above  shortcomings,  we propose several  reinforcement learning algorithms for a class of multi-agent control problems  called \emph{deep structured teams}, introduced in~\cite{Jalal2019MFT,Jalal2019risk, Jalal2019Automatica,Jalal2020Nash,Vida2020CDC,Masoud2020CDC}, where  the interactions between the  decision makers  are modelled  by  a number of linear regressions (weighted averages) of states and actions, which is similar to  the interactions between the neurons of a feed-forward deep neural network.  In general, deep structured teams are decentralized control systems whose solutions are  amenable to the size of the problem. More precisely, the complexity of finding an optimal  solution of Markov chain deep structured team is polynomial (rather than exponential) with respect to the number of agents and is linear (rather than exponential) with respect to the control horizon~\cite{Jalal2019MFT}.  On the other hand,  the complexity of a linear quadratic deep structured team is independent of the number of agents~\cite{Jalal2019risk}.
 It is worth highlighting  that deep structured teams are the generalization of the notion of mean-field teams initially  introduced in~\cite{arabneydi2016new} and showcased in~\cite{JalalCDC2017,JalalCDC2018,
 Jalal2019LCSS,JalalCCECE2018,JalalACC2018,Jalal2017linear,JalalCDC2015,JalalACC2019}.   

The remainder of the paper is organized as follows. In Section~\ref{sec:formulation},  two models of  deep structured teams are formulated,  one with  finite state and action spaces and the other one with  infinite spaces. In Section~\ref{sec:planning}, different methods  are discussed for solving the planning problem for the case when the model is known. In Section~\ref{sec:RL}, some reinforcement learning methods are presented  for the case when the model is not known.  An example of a smart grid is provided in Section~\ref{sec:numerical} to verify the effectiveness of the proposed algorithms,  and the paper is then concluded  in Section~\ref{sec:conclusions}.

%

\section{Problem Formulation}\label{sec:formulation}

In this  paper, $\mathbb{N}$,  $\mathbb{R}$ and  $\mathbb{R}_{\geq 0}$  are  the sets of  natural numbers, real numbers and  non-negative real numbers, respectively.   For any $k, t \in \mathbb{N}$,  $\mathbb{N}_k$  denotes the finite set $\{1,\ldots,k\}$ and  $x_{1:t}$  denotes the vector $(x_1,\ldots,x_t)$.  For any vectors $x$, $y$, and $z$,  short-hand notation $\VEC(x, y, z)$  denotes the vector $[x^\intercal, y^\intercal, z^\intercal]^\intercal$.  For any matrices $A$, $B$  and $C$  that have the same number of columns, $\ROW(A, B,C)$ denotes  the  matrix $[A^\intercal, B^\intercal, C^\intercal]^\intercal$.
 Given any square matrices $A$, $B$ and $C$, $\DIAG(A, B, C)$  denotes the block diagonal  matrix with matrices $A$, $B$ and $C$ on its main diagonal.  In addition, $\Prob{\boldsymbol \cdot}$ is the probability of a random variable, $\Exp{\boldsymbol \cdot}$ is the expectation of an event, $\ID{\boldsymbol \cdot}$ is the indicator function of a set, $\TR(\boldsymbol \cdot)$ is the trace of a matrix, $\Ninf{\boldsymbol \cdot}$  is the infinity norm of a vector, and  $| \boldsymbol \cdot|$ is the absolute value of a real number or the cardinality of a set.    For any $n \in \mathbb{N}$,   $\Binopdf(n,p)$ denotes the binomial probability distribution of $n$ trials with success probability~$p \in [0,1]$.  For any finite  set $\mathcal{X}$,  the space of probability measures on $\mathcal{X}$ is denoted by: 
$\mathcal{P}(\mathcal{X})=\{(a_1,\ldots,a_{|\mathcal{X}|}) \big| a_i \in [0,1], i \in \mathbb{N}_{|\mathcal{X}|}, \sum_{i=1}^{|\mathcal{X}|} a_i=1\}$.
Furthermore,  the space of empirical distributions  over $\mathcal{X}$ with $n$ samples is given by: 
$ \Emp_n(\mathcal{X})=\{(a_1,\ldots,a_{|\mathcal{X}|}) \big| a_i \in \{0, \frac{1}{n},\ldots,1\}, i \in \mathbb{N}_{|\mathcal{X}|}, \sum_{i=1}^{|\mathcal{X}|} a_i=1\}$, 
where $\Emp_n (\mathcal{X}) \subset \mathcal{P}(\mathcal{X})$.

Consider a stochastic dynamic control system consisting of $n \in \mathbb{N}$ agents (decision makers). Let $x^i_t \in \mathcal{X}$, $u^i_t \in \mathcal{U}$ and $w^i_t \in \mathcal{W}$ denote the state, action and noise of agent $ i \in \mathbb{N}_n$ at time $t \in \mathbb{N}$. In addition, define  $\mathbf x_t=\VEC(x^1_t,\ldots,x^n_t) \in \mathcal{X}^n$, $\mathbf u_t=\VEC(u^1_t,\ldots,u^n_t) \in \mathcal{U}^n$, and $\mathbf w_t=\VEC(w^1_t,\ldots,w^n_t) \in \mathcal{W}^n$.   The initial states $\mathbf x_1$ and  noises  $\mathbf w_t$, $t \in \mathbb{N}$, are distributed randomly  with respect to  joint probability distribution functions $P_{\mathbf X}$ and  $P_{\mathbf W}$, respectively. It is assumed that random variables $\{\mathbf x_1, \mathbf{w}_1,\ldots,\mathbf w_T\}$, $T \in \mathbb{N}$, are  defined on a common probability space and are mutually   independent across any control horizon $T$.

 In this paper, we consider two fundamental models in deep structured teams, where the  first one has a finite space with a controlled  Markov chain formulation and the second one has an infinite space with  a linear quadratic structure.   
 
\subsection*{Model I: Finite-valued sets}
Let spaces $\mathcal{X}$, $\mathcal{U}$ and $\mathcal{W}$ be finite sets. Denote by $\mathfrak{D}_t \in \Emp_n(\mathcal{X} \times \mathcal{U})$ the empirical distribution of states and  actions  at time $t \in \mathbb{N}$, i.e.
\begin{equation}\label{eq:def_joint_distribution}
\mathfrak{D}_t(x,u):=\frac{1}{n}\sum_{i=1}^n \ID{x^i_t=x} \ID{u^i_t=u}, \quad \forall x \in \mathcal{X}, u \in \mathcal{U}. 
\end{equation}
Let $d_t$  denote  the  empirical distribution of states  at time $t $: 
\begin{equation}\label{eq:def_distribution of states}
d_t(x):= \sum_{u \in \mathcal{U}}\mathfrak{D}_t(x,u)= \frac{1}{n}\sum_{i=1}^n \ID{x^i_t=x}, \quad x \in \mathcal{X}.
\end{equation}
The agents are coupled in dynamics through $\mathfrak{D}_t$,  representing the  aggregate behaviour of agents at $t \in \mathbb{N}$.  More precisely, the state of agent $i \in \mathbb{N}$ at time $t \in \mathbb{N}$ evolves as follows:
\begin{equation}\label{eq:dynamics}
x^i_{t+1}=f(x^i_t, u^i_t, \mathfrak{D}_t,w^i_t),
\end{equation}
where $f: \mathcal{X} \times \mathcal{U} \times \Emp_n(\mathcal{X} \times \mathcal{U})  \times \mathcal{W} \rightarrow \mathcal{X}$, and $\{w^i_t\}_{t=1}^\infty$ is an i.i.d. random  process with probability mass function $P_W$, i.e., $P_{\mathbf W}(\mathbf w_t)=\prod_{i=1}^n P_W(w^i_t)$. Alternatively, the dynamics~\eqref{eq:dynamics} may be written in terms of transition probability matrix:
\begin{multline}
p(x^i_{t+1}, x^i_t, u^i_t, \mathfrak{D}_t):=\Prob{x^i_{t+1} \mid x^i_t,u^i_t, \mathfrak{D}_t}\\
=\sum_{w \in \mathcal{W}} \ID{x^i_{t+1}=f(x^i_t, u^i_t, \mathfrak{D}_t,w)}P_W(w),
\end{multline}
where $\sum_{y \in \mathcal{X}} p(y, x^i_t, u^i_t,\mathfrak{D}_t)=1$.  

 The agents are also coupled in cost function; this is formulated by defining the per-step cost function    $c(x^i_t,u^i_t,\mathfrak{D}_t)$ at  $t \in \mathbb{N}$, where    $c: \mathcal{X} \times \mathcal{U} \times \Emp_n(\mathcal{X} \times \mathcal{U}) \rightarrow \mathbb{R}_{\geq 0}$ and $i \in \mathbb{N}_n$. 
\subsection*{Model II: Infinite-valued sets}
Let  $\mathcal{X}=\mathbb{R}^{h_x}$, $\mathcal{U}=\mathbb{R}^{h_u}$ and $\mathcal{W}=\mathbb{R}^{h_x}$ be  finite-dimensional Euclidean spaces, where $h_x,h_u \in \mathbb{N}$.  Let also the weight (impact factor) $\alpha^{i,j} \in \mathbb{R}$ denote  the influence of agent $i \in \mathbb{N}_n$ on  the $j$-th feature, $j \in \mathbb{N}_z$, $z  \in \mathbb{N} $, where these  impact factors are assumed to be orthonormal  vectors  in the feature space, i.e., 
\begin{equation}
\frac{1}{n} \sum_{i =1}^n \alpha^{i,j} \alpha^{i,k}=\ID{j=k}, \quad j,k \in \mathbb{N}_z.
\end{equation}
For feature $j \in \mathbb{N}_z$, define the following linear regressions:
\begin{align}
\bar x^j_t:=\frac{1}{n}\sum_{i=1}^n \alpha^{i,j} x^i_t, \quad \bar u^j_t:=\frac{1}{n}\sum_{i=1}^n \alpha^{i,j} u^i_t.
\end{align}
 Define $\bar {\mathbf x}_t:=\VEC(\bar x^1_t,\ldots,\bar x^z_t)$ and  $\bar {\mathbf u}_t:=\VEC(\bar u^1_t,\ldots,\bar u^z_t)$. The dynamics of agent $i \in \mathbb{N}_n$ is coupled  with  other agents through $\bar{\mathbf x}_t$ and $\bar{\mathbf u}_t$ as follows:
\begin{equation}\label{eq:dynamics_LQ}
x^i_{t+1}=Ax^i_t+Bu^i_t+ \sum_{j=1}^z \alpha^{i,j}(\bar A^j \bar{\mathbf x}_t +\bar B^j \bar{\mathbf u}_t)+w^i_t,
\end{equation}
where matrices $A$, $B$, $\bar A^j$ and $\bar B^j$, $j \in \mathbb{N}_z$,  have appropriate dimensions.  To have a well-posed problem,  it is  assumed that the  initial states and  driving noises  have  uniformly bounded  covariance matrices with respect to time, and that the set of admissible
control actions are square
integrable for all agents.  For simplicity of presentation,  it is also assumed that the  initial states and  local noises have zero mean.    The per-step cost function of agent $i \in \mathbb{N}_n$ is defined as:
\begin{equation}\label{eq:cost_LQ}
c(x^i_t,u^i_t, \bar{\mathbf x}_t, \bar{\mathbf u}_t )=(x^i_t)^\intercal Q x^i_t + (u^i_t)^\intercal R u^i_t
+\bar{\mathbf x}_t^\intercal \bar Q \bar{\mathbf x}_t+ \bar{\mathbf u}_t^\intercal \bar{R} \bar{\mathbf u}_t,
\end{equation}
where   $Q$, $R$, $\bar{Q}$ and $\bar{R}$  are symmetric matrices with appropriate dimensions.   Let
\begin{align}
\bar{\mathbf A}&:=\DIAG(A,\ldots,A)+\ROW[\bar A^1,\ldots,\bar A^z],\\
\bar{\mathbf B}&:=\DIAG(B,\ldots,B)+\ROW[\bar B^1,\ldots,\bar B^z],\\
\bar{\mathbf Q}&:=\DIAG(Q,\ldots,Q)+\bar Q,\quad 
\bar{\mathbf R}:=\DIAG(R,\ldots,R)+\bar R.
\end{align}

\begin{remark}[Weakly coupled]\label{remark:equivariant}
\emph{Consider  a  special case where agents are coupled in  the dynamics~\eqref{eq:dynamics_LQ} and cost function~\eqref{eq:cost_LQ} through the following  terms:  
\begin{equation}
\sum_{j=1}^z \alpha^{i,j} (\bar A^j \bar x^j_t + \bar B^j \bar{u}^j_t), \quad \sum_{j=1}^z (\bar x^j_t)^\intercal \bar Q^j \bar x_t^j + (\bar u^j_t)^\intercal \bar R^j \bar u_t^j.
\end{equation}
The above case arises in various systems specially those with equivariant structure~\cite[Propositions 3 \& 4]{Jalal2019risk}.}
\end{remark}

\subsection{Information structure}\label{sec:info}
Following the terminology of deep structured  teams~\cite{Jalal2019MFT,Jalal2019risk}, we  refer to  the aggregate state of agents as  \emph{deep state},  which is  the empirical distribution $d_t$ in  Model I and weighted average $\bar{\mathbf x}_t$ in Model~II. The first information structure considered in this paper is called \emph{deep state sharing} (DSS).  Under this information structure,    the action of agent $i \in \mathbb{N}_n$ at time $t \in \mathbb{N}$ in Model~I  is chosen  with respect to  a probability mass function $g_t(\boldsymbol \cdot \mid x^i_t, d_t)  \in \mathcal{P}(\mathcal{U})$, i.e.,
\begin{equation}
u^i_t \sim g_t(\boldsymbol \cdot \mid x^i_t, d_t). \tag{Model I: DSS}
\end{equation} 
For Model II, however,  the action of agent $i \in \mathbb{N}_n$ at time $t \in \mathbb{N}$ is selected  with respect to a probability distribution function $g^i_t(\boldsymbol \cdot \mid x^i_t, \bar{\mathbf x}_t)  \in \mathcal{P}(\mathcal{U})$, i.e., 
\begin{equation}
u^i_t \sim g^i_t(\boldsymbol \cdot \mid x^i_t, \bar{\mathbf x}_t). \tag{Model II: DSS}
\end{equation} 

In practice, there are various methods  to  share the deep state among agents. For example, one may use a central distributor (such as a cloud-based server)  to collect the states, compute the deep state  and  broadcast it  among the agents. Alternatively,  one can utilize distributed techniques such as consensus-based algorithms  based on the    local interaction of each agent  with its neighbours.  However, it is sometimes  infeasible to share the deep state, specially  when the number of agents is very large. In such a case,  we consider  another  information structure  called  \emph{no sharing} (NS), where 
\begin{equation}
u^i_t \sim g_t(\boldsymbol \cdot \mid x^i_t) \in \mathcal{P}(\mathcal{U}), \tag{Model I: NS}
\end{equation}
and 
\begin{equation}
u^i_t \sim g^i_t(\boldsymbol \cdot \mid x^i_t) \in \mathcal{P}(\mathcal{U}). \tag{Model II: NS}
\end{equation}
It is to be noted that the strategy of agent $i \in \mathbb{N}_n$ in  Model~II  depends on  its index~$i$ as well as its local state $x^i_t$.
\subsection{Objective function}
 Let  $\mathbf g:=\{g_{t}\}_{t=1}^\infty$ denote  the  control strategy for the system.
Two performance indexes are considered, namely,  discounted cost and time-average cost.  In particular, given any discount factor  $\beta \in (0,1)$, the objective function of Model I is defined as follows:
\begin{equation}
J_{n,\beta}^{\text{(I)}}(\mathbf g):= (1-\beta) \mathbb{E}^{\mathbf g}[\frac{1}{n}\sum_{i=1}^n \sum_{t=1}^\infty \beta^{t-1} c(x^i_t,u^i_t, \mathfrak{D}_t)].
\end{equation}
Similarly,   the following  discounted cost function is defined  for Model~II:
\begin{equation}
J_{n,\beta}^{\text{(II)}}(\mathbf g):= (1-\beta) \mathbb{E}^{\mathbf g}[\frac{1}{n}\sum_{i=1}^n \sum_{t=1}^\infty \beta^{t-1} c(x^i_t,u^i_t,  \bar{\mathbf x}_t, \bar{\mathbf u}_t)].
\end{equation} 
In this article,  standard  mild assumptions are imposed on the model  to ensure that   the  total cost  is always bounded. As a result,    it is possible  to obtain the  time-average cost function as the limit of the discounted cost function. More precisely, the following holds for Model I:
\begin{align}
&J_{n,1}^{\text{(I)}}(\mathbf g):=\lim_{\beta \rightarrow 1} J^{\text{(I)}}_{n,\beta}(\mathbf g)\\
&= \lim_{\beta \rightarrow 1} \frac{ \lim_{T \rightarrow \infty}\Exp{\frac{1}{n}\sum_{i=1}^n \sum_{t=1}^T \beta^{t-1} c(x^i_t,u^i_t, \mathfrak{D}_t}}{\lim_{T \rightarrow \infty}\Exp{\sum_{t=1}^T \beta^{t-1}}}\\
&=\limsup_{T \rightarrow \infty} \frac{ \lim_{\beta \rightarrow 1}\Exp{\frac{1}{n}\sum_{i=1}^n \sum_{t=1}^T \beta^{t-1} c(x^i_t,u^i_t, \mathfrak{D}_t)}}{\lim_{\beta \rightarrow 1}\Exp{\sum_{t=1}^T \beta^{t-1}}}\\
 &=\limsup_{ T \rightarrow \infty}  \frac{1}{T} \mathbb{E}^{\mathbf g}[\frac{1}{n}\sum_{i=1}^n \sum_{t=1}^T  c(x^i_t,u^i_t, \mathfrak{D}_t)].
\end{align}
Analogously, one has the following for Model II:
\begin{align}
&J_{n,1}^{\text{(II)}}(\mathbf g):=\lim_{\beta \rightarrow 1} J^{\text{(II)}}_{n,\beta}(\mathbf g)\\
&= \lim_{\beta \rightarrow 1} \frac{ \lim_{T \rightarrow \infty}\Exp{\frac{1}{n}\sum_{i=1}^n \sum_{t=1}^T \beta^{t-1} c(x^i_t,u^i_t,  \bar{\mathbf x}_t, \bar{\mathbf u}_t)}}{\lim_{T \rightarrow \infty}\Exp{\sum_{t=1}^T \beta^{t-1}}}\\
&= \limsup_{T \rightarrow \infty}  \frac{1}{T} \mathbb{E}^{\mathbf g}[\frac{1}{n}\sum_{i=1}^n \sum_{t=1}^T   c(x^i_t,u^i_t,  \bar{\mathbf x}_t, \bar{\mathbf u}_t)].
\end{align}

It is to be noted that  there is no  general theory for infinite-space average cost functions.
\subsection{Problem statement}
Four problems are investigated.

\begin{problem}[Planning with DSS]\label{problem1}
Given dynamics $\{f, A, B, \bar{\mathbf A}, \bar{\mathbf B} \}$, cost function $\{c, Q, R, \bar{\mathbf Q}, \bar{\mathbf R}\}$, number of agents $n$,  probability distribution function of the initial states $P_{\mathbf X}$, probability distribution function of noises $P_{\mathbf W}$ and discount  factor $\beta \in (0,1]$,  find the optimal strategy $\mathbf g^\ast$ such that under DSS information structure for every strategy~$\mathbf g$:
\begin{equation}
J^{(\ell)}_{n,\beta}(\mathbf g^\ast) \leq J^{(\ell)}_{n,\beta}(\mathbf g), \quad  \ell \in \{\text{I}, \text{II}\}.
\end{equation}
\end{problem}

\begin{problem}[Planning with NS]\label{problem2}
Given dynamics $\{f, A, B, \bar{\mathbf A}, \bar{\mathbf B} \}$, cost function $\{c, Q, R, \bar{\mathbf Q}, \bar{\mathbf R}\}$, number of agents $n$,  probability distribution function of the initial states $P_{\mathbf X}$, probability distribution function of noises $P_{\mathbf W}$ and discount factor $\beta \in (0,1]$,  find a sub-optimal strategy $\hat{\mathbf g}$ such that under NS information structure for every strategy~$\mathbf g$:
\begin{equation}
J^{(\ell)}_{n,\beta}(\hat{\mathbf g}) \leq J^{(\ell)}_{n,\beta}(\mathbf g)+\varepsilon(n), \quad  \ell \in \{\text{I}, \text{II}\},
\end{equation}
where $\lim_{n \rightarrow \infty} \varepsilon(n)=0$.
\end{problem}

\begin{problem}[Reinforcement learning with DSS]\label{problem3}
Given state and action spaces $\{\mathcal{X}, \mathcal{U}\}$ and discount factor $\beta \in (0,1]$,  develop a reinforcement learning  algorithm whose performance under the learned strategy $\mathbf g_k$, $k \in \mathbb{N}$, converges to  that under the optimal strategy $\mathbf g^\ast$, as the number of iterations~$k$ increases.
\end{problem}

\begin{problem}[Reinforcement learning with NS]\label{problem4}
Given state and action spaces $\{\mathcal{X}, \mathcal{U}\}$ and discount  factor $\beta \in (0,1]$,  develop a reinforcement learning  algorithm whose performance under the learned strategy $\hat{\mathbf g}_k$, $k \in \mathbb{N}$, converges to that under the sub-optimal strategy $\hat{\mathbf g}$, as the number of iterations $k $ increases.
\end{problem}

\section{Main results for Problems~\ref{problem1} and~\ref{problem2}}\label{sec:planning}

Following~\cite{Jalal2019MFT}, we define a local control law $\gamma_t: \mathcal{X} \rightarrow \mathcal{P}(\mathcal{U})$  for Model I such that   under DSS information structure,
\begin{equation}\label{eq:DSS_gamma}
\gamma_t:=g_t(\boldsymbol \cdot, d_t), 
\end{equation}
and under NS information structure,
\begin{equation}\label{eq:NS_gamma}
\gamma_t:=g_t(\boldsymbol \cdot).
\end{equation}
From the definition of DSS and NS strategies given  in Subsection~\ref{sec:info} and  the change of variable introduced in~\eqref{eq:DSS_gamma} and~\eqref{eq:NS_gamma}, it follows that  the  action  of agent $i$ is  selected  randomly  with respect to the  probability mass function $\gamma_t(x^i_t)$, i.e.
\begin{equation}
u^i_t \sim \gamma_t(x^i_t).
\end{equation}
\begin{lemma}
Given any $n \in \mathbb{N}$,  $d_t \in \Emp_n(\mathcal{X})$ and $\gamma_t: \mathcal{X} \rightarrow \mathcal{P}(\mathcal{U})$ at time $t \in \mathbb{N}$, the following relations hold:
\begin{equation}
n \mathfrak{D}_t(x,u)  \sim \Binopdf(n d_t(x),\gamma_t(x)(u)), \quad x \in \mathcal{X},\hspace{.1cm} u \in \mathcal{U},
\end{equation}
and 
\begin{equation}
\Exp{\mathfrak{D}_t(x,u)}= d_t(x) \gamma_t(x)(u), \quad x \in \mathcal{X}, \hspace{.1cm} u \in \mathcal{U}.
\end{equation}
\end{lemma}
\begin{proof}
The proof directly follows from~\eqref{eq:def_joint_distribution} such that 
\begin{equation}
n \mathfrak{D}_t(x,u)=\sum_{i=1}^n \ID{x^i_t=x} \ID{u^i_t=u},
\end{equation}
where the above equation consists of $n d_t(x)$ independent  binary random variables with success probability $\gamma_t(x)(u)$. 
\end{proof}

To ease the exposition of deep Chapman-Kolmogorov equation introduced in~\cite{Jalal2019MFT},   define
\begin{equation}
\mathcal{T}(x^i_{t+1},x^i_t,\gamma_t,d_t):=\sum_{u \in \mathcal{U}} \Prob{x^i_{t+1}\mid x^i_t, u, \gamma_t, d_t} \gamma_t(x^i_t)(u).
\end{equation}
Given any 
 $x,x' \in \mathcal{X}$, $\gamma: \mathcal{X} \rightarrow \mathcal{P}(\mathcal{U})$ and $d \in \Emp_n(\mathcal{X}) $, define the vector-valued function $\phi(x',x,\gamma, d) \in  \mathcal{P}\left(\{0,1,\ldots, n  d(x)\}\right)$ such that
\begin{multline}\label{eq:phi}
\phi(x',x,\gamma, d):=  \delta_0(n  d(x))  \\
 +\ID{d(x) > 0} \Binopdf{\left( n  d(x), \mathcal{T}(x',x,\gamma,d) \right)},
\end{multline} 
where $\delta_0(n d(x)) $ is the Dirac measure with the domain set $\{0,1,\ldots,$ $ n  d(x)\}$ and  a unit mass concentrated at zero.   In addition, let $\bar \phi(x',\gamma,d) \in  \mathcal{P}\left(\{0,1,\ldots, n\}\right)$ be the convolution function  of $\phi(x',x,\gamma, d)$ over all states $x \in \mathcal{X}=:\{s_1,\ldots,s_{|\mathcal{X}|}\}$, i.e.,
\begin{equation}\label{eq:tilde-phi}
\bar \phi(x',\gamma, d):= \phi(x',s_1,\gamma, d) \ast \ldots *\phi(x',s_{|\mathcal{X}|},\gamma, d),
\end{equation}
where $\bar \phi(x',\gamma, d)$ is a vector of size $n+1$. 
\begin{lemma}[Deep Chapman-Kolmogorov equation~\cite{Jalal2019MFT}]\label{thm:mean_field_iid}
 Given $d_t  $ and $\gamma_t$ at time $t \in \mathbb{N}$, the transition probability matrix of the  deep state can be computed  as follows:  for any $x' \in \mathcal{X}$ and $y \in \{0,1,\ldots,n\}$,
\begin{equation}
\Prob{d_{t+1}(x')=\frac{y}{n} \mid d_t,\gamma_t}= \bar \phi(x',\gamma_t, d_t)(y+1).
\end{equation}
\end{lemma}
\begin{proof}
The proof follows from~\eqref{eq:def_distribution of states} and the fact that the probability distribution of the  sum of independent random variables can be described by the convolution of their individual  probability distributions. See~\cite[Theorem 3]{Jalal2019MFT} and~\cite[Theorem 1]{Jalal2019TSNE} for more details.
\end{proof}

 We now define a non-standard Bellman equation for Model~I such that for every $d_t \in \Emp_n(\mathcal{X})$ and $t \in \mathbb{N}$:
 \begin{multline}\label{eq:bellman_exact_n}
 V(d_t)=\min_{\gamma_t} (\bar c(d_t,\gamma_t) + \beta \Exp{V(d_{t+1}) \mid d_t,\gamma_t})  \\
 =\min_{\gamma_t} (\bar c(d_t,\gamma_t) + \beta  \sum_{\tilde d \in \Emp_n(\mathcal{X})} \Prob{d_{t+1}=\tilde d \mid d_t,\gamma_t} V(\tilde d)), 
 \end{multline}
 where
\begin{align}
\bar  c(d_t,\gamma_t)&:=\Exp{\frac{1}{n} \sum_{i=1}^n c(x^i_t,u^i_t, \mathfrak{D}_t) \mid d_t, \gamma_t}\\
 &\quad = \Exp{\sum_{x,u} c(x,u,\mathfrak{D}_t) \mathfrak{D}_t(x,u)   \mid d_t, \gamma_t}\\
 &\quad = \sum_{\mathfrak{D}} (\sum_{x,u} c(x,u,\mathfrak{D})\mathfrak{D}(x,u)) \Prob{\mathfrak{D}_t=\mathfrak{D} \mid d_t,\gamma_t}. 
\end{align}

\begin{theorem}\label{thm:theorem1}
Let $\psi^\ast (d)$, $d \in \Emp_n(\mathcal{X})$,  be a minimizer of the right-hand side of equation~\eqref{eq:bellman_exact_n}. The following strategy is an optimal solution for Problem~\ref{problem1} with  Model I:
\begin{equation}
u^{i,\ast}_t \sim g^\ast(x^i_t, d_t):= \psi^\ast(d_t)(x^i_t).
\end{equation}
\end{theorem}
\begin{proof}
The proof follows from the fact that $d_t$ is an information state under  strategy $\gamma_t$. For more details, see~\cite[Theorem~2]{Jalal2019MFT}.
\end{proof}


To find the solution of Problem~\ref{problem1} with  Model II,  we impose the following standard assumption.

\begin{assumption}\label{ass: stability_known_LQ}
Let  matrices $Q$ and $\bar{\mathbf Q}$  be  positive semi-definite, and  matrices $R$ and $\bar{\mathbf R}$  be  positive definite. In addition,
let  $(A, B)$ and $( \bar{\mathbf A}, \bar{\mathbf{B}})$ be  stabilizable, and  $( A,Q^{1/2})$ and $( \bar{\mathbf A},  \bar{\mathbf Q}^{1/2})$ be detectable.
\end{assumption}

We now describe  the algebraic form of  the deep Riccati equation introduced in~\cite{Jalal2019risk,arabneydi2016new}:  
\begin{equation}\label{eq:deep_Riccati}
\begin{cases}
P= Q+\beta A^\intercal P A-\beta A^\intercal P B (B^\intercal P B+\beta^{-1}R)^{-1}B^\intercal PA,\\
\bar{\mathbf P}= \bar{\mathbf Q}+\beta \bar{\mathbf A}^\intercal \bar{\mathbf P} \bar{\mathbf A}-\beta \bar{\mathbf A}^\intercal \bar{\mathbf P} \bar{\mathbf B} (\bar{\mathbf B}^\intercal \bar{\mathbf P} \bar{\mathbf B}+\beta^{-1}\bar{\mathbf R})^{-1}\bar{\mathbf B}^\intercal \bar{\mathbf P} \bar{\mathbf A}.
\end{cases}
\end{equation}
Define  the following feedback gains:
\begin{equation}
\begin{cases}
\theta^\ast :=-(B^\intercal P B+\beta^{-1}R)^{-1}B^\intercal PA,\\
\bar{\boldsymbol \theta}^\ast:= -  (\bar{\mathbf B}^\intercal \bar{\mathbf P} \bar{\mathbf B}+\beta^{-1}\bar{\mathbf R})^{-1}\bar{\mathbf B}^\intercal \bar{\mathbf P} \bar{\mathbf A}.
\end{cases}
\end{equation}

\begin{remark}\label{remark:equivariant_Riccati}
\emph{For the weakly coupled case in Remark~\ref{remark:equivariant},   $\bar{\mathbf P}$ decomposes into  $z \in \mathbb{N}$ smaller Riccati equations such that for any $ j \in \mathbb{N}_z$:
\begin{multline}\label{eq:deep_Riccati_E}
\bar P^j= Q+\bar Q^j+\beta (A+\bar A^j)^\intercal \bar P^j (A+\bar A^j)\\
-\beta (A+\bar A^j)^\intercal \bar P^j (B+\bar B^j) ((B+\bar B^j)^\intercal \bar P^j (B+\bar B^j)+\beta^{-1}(R+\bar R^j))^{-1}\\
\times(B+\bar B^j)^\intercal \bar P^j (A+\bar A^j),
\end{multline}
where $\bar{\mathbf P}=\DIAG(\bar P^1,\ldots,\bar P^z)$.}
\end{remark}

\begin{remark}
\emph{Note that the dimension of the deep  Riccati equation~\eqref{eq:deep_Riccati} does not depend on the number of agents $n$; however,  it depends on the number of orthonormal features~$z$. }
\end{remark}

\begin{theorem}\label{thm:LQ_DSS}
Let Assumption~\ref{ass: stability_known_LQ} hold. The  optimal solution of Problem~\ref{problem1}  for Model II is  described by:
\begin{equation}\label{eq:optimal_LO}
u^{i,\ast}_t=\theta^\ast x^i_t+ \sum_{j=1}^z \alpha^{i,j} (\bar \theta^{j,\ast} \bar{\mathbf x}_t - \theta^{\ast} \bar x^j_t), \quad t \in \mathbb{N},
\end{equation}
where $\bar{\boldsymbol \theta}^\ast=:\ROW(\bar \theta^{1,\ast},\ldots,\bar \theta^{z,\ast})$.
For the special case of weakly coupled case in Remark~\ref{remark:equivariant}, one has:
\begin{equation}
u^{i,\ast}_t=\theta^\ast x^i_t+ \sum_{j=1}^z \alpha^{i,j} (\bar \theta^{j,\ast}  - \theta^{\ast}) \bar x^j_t.  
\end{equation}
\end{theorem}
\begin{proof}
The proof follows from a change of variables and certainty equivalence principle. Define
\begin{equation}
\Delta x^i_t:=x^i_t - \sum_{j=1}^z \alpha^{i,j} \bar x^j_t, \quad \Delta u^i_t:=u^i_t - \sum_{j=1}^z \alpha^{i,j} \bar u^j_t. 
\end{equation}
The above change of variables is  a gauge transformation, initially introduced in~\cite{arabneydi2016new} and extended in~\cite{Jalal2019risk} to optimal control systems and in~\cite{Jalal2019Automatica} to dynamic games. From  the orthogonality induced by the above gauge transformation,   the dynamics~\eqref{eq:dynamics_LQ} and cost function~\eqref{eq:cost_LQ} can be decomposed  into $n$ identical linear quadratic problems with states and actions $(\Delta x^i_t, \Delta u^i_t)$, $i \in \mathbb{N}_n$, and one linear quadratic problem with  state and action $(\bar{\mathbf x}_t, \bar{\mathbf u}_t)$. The deep Riccati equation~\eqref{eq:deep_Riccati} gives the solution of the  above problems. For the special case of weakly coupled systems, the Riccati equation associated with  the state and action $(\bar{\mathbf x}_t, \bar{\mathbf u}_t)$ decomposes into $z$ smaller Riccati equations  given by~\eqref{eq:deep_Riccati_E}.    Note that  Assumption~\ref{ass: stability_known_LQ} implies that  for any $\beta \in (0,1]$,  $(\sqrt{\beta}A, \sqrt \beta B)$ and $(\sqrt \beta \bar{\mathbf A}, \sqrt \beta \bar{\mathbf{B}})$ are stabilizable, and  $(\sqrt \beta A,Q^{1/2})$ and $(\sqrt \beta \bar{\mathbf A},  \bar{\mathbf Q}^{1/2})$ are detectable. As a result, algebraic Riccati equation~\eqref{eq:deep_Riccati} has a unique, bounded and positive solution. See~\cite{Jalal2019risk,arabneydi2016new} for more details.
\end{proof}

\begin{remark}[Extended cost]
\emph{The main results of this paper naturally extend to cost functions with  post-decision states, i.e.,  $c(x^i_t,u^i_t, \mathfrak{D}_t, x^i_{t+1}, d_{t+1})$ and $c(x^i_t,u^i_t, \bar{\mathbf x}_t, \bar{\mathbf u}_t, x^i_{t+1},\bar{\mathbf x}_{t+1} )$. In addition, it is straightforward  to consider cross terms between the states and actions in~\eqref{eq:cost_LQ}.}
\end{remark}

\begin{remark}[Generalization]
\emph{The main results of this paper naturally generalize to state-dependent discount factor and  finite-horizon cost functions  as well as  multiple sub-populations with partially deep state sharing  information structure~\cite{Jalal2019risk} and  intermittent deep state sharing~\cite{Jalal2019LCSS}.}
\end{remark}

\subsection{NS information structure: mean-field approximation }
When deep state is not observed,  the above solutions are not practical.  To overcome this shortcoming, one can   approximate the deep state  by  mean field approximation~\cite{parisi1988statistical},   where
the strong law of large numbers provides a simple  asymptotic estimate.

Let $\mathfrak{M}_t \in \mathcal{P}(\mathcal{X} \times \mathcal{U})$ and $m_t \in \mathcal{P}(\mathcal{X})$  denote   the mean-field approximations of  $\mathfrak{D}_t \in \Emp_n(\mathcal{X} \times \mathcal{U})$ and $d_t \in \Emp_n(\mathcal{X})$, respectively, i.e.,
\begin{align}
\mathfrak{M}_t(x,u)&=\lim_{n \rightarrow \infty} \mathfrak{D}_t(x,u)=\lim_{n \rightarrow \infty}\frac{1}{n} \sum_{i=1}^n \ID{x^i_t=x} \ID{u^i_t=u}, \nonumber \\
m_t(x)&=\lim_{n \rightarrow \infty} d_t(x)=\lim_{n \rightarrow \infty}\frac{1}{n} \sum_{i=1}^n \ID{x^i_t=x}.
\end{align}
 We now define a non-standard Bellman equation for Model~I such that for every $m_t \in \mathcal{P}(\mathcal{X})$ and $t \in \mathbb{N}$:
 \begin{align}\label{eq:bellman_exact_inf}
 \hat V(m_t)&=\min_{\gamma_t} ( \hat c(m_t,\gamma_t) + \beta \Exp{\hat V(m_{t+1}) \mid m_t,\gamma_t}) \nonumber  \\
 &=\min_{\gamma_t} (\hat c(m_t,\gamma_t) + \beta   \hat V(\hat f(m_t,\gamma_t))), 
 \end{align}
 where
\begin{align}
\mathfrak{M}_t(x,u)&=m_t(x) \gamma_t(x)(u), \quad x \in \mathcal{X}, u \in \mathcal{U},\\
\hat  c(m_t,\gamma_t)&:=  \sum_{x,u} c(x,u,\mathfrak{M}_t)m_t(x) \gamma_t(x)(u),
\end{align}
and for every $x' \in \mathcal{X}$,
\begin{align}\label{eq:mean_field_dynamics}
m_{t+1}(x')&=\sum_{x \in \mathcal{X}} \sum_{u \in \mathcal{U}} m_t(x) \gamma_t(x)(u) \Prob{x'|x,u,\mathfrak{M}} \nonumber\\
&=:\hat f(m_t, \gamma_t)(x').
\end{align}

\begin{assumption}\label{ass:iid_x}
The initial states $\mathbf x_1$ are i.i.d. random variables with probability mass function $P_X$.
\end{assumption}

\begin{assumption}\label{ass:Lipschitz_finite}
For any $x,x' \in \mathcal{X}$, $u \in \mathcal{U}$ and $\mathfrak{M}_1, \mathfrak{M}_2 \in \mathcal{P}(\mathcal{X} \times \mathcal{U})$, there exist positive real constants $H^p$ and $H^c$ (that do not depend on $n$) such that
\begin{align}
|\Prob{x'|x,u,\mathfrak{M}_1} - \Prob{x'|x,u,\mathfrak{M}_2}| \leq H^p \Ninf{\mathfrak{M}_1 - \mathfrak{M}_2},\\
|c(x,u,\mathfrak{M}_1)-c(x,u,\mathfrak{M}_2) |\leq H^c \Ninf{\mathfrak{M}_1 - \mathfrak{M}_2}.
\end{align}
\end{assumption}
The above assumption is mild because any polynomial function of $\mathfrak{M} \in \mathcal{P}(\mathcal{X}\times \mathcal{U})$ is  Lipschitz on $\mathfrak{M}$ since  $\mathcal{P}(\mathcal{X}\times \mathcal{U})$ is a confined space.
\begin{theorem}\label{thm:theorem 3}
Let Assumptions~\ref{ass:iid_x} and~\ref{ass:Lipschitz_finite} hold. Let  also $\hat \psi (m)$, $m \in \mathcal{P}(\mathcal{X})$,  be a minimizer of the right-hand side of equation~\eqref{eq:bellman_exact_inf}. The following strategy is a solution for Problem~\ref{problem2} with Model I:
\begin{equation}
\hat u^{i}_t \sim \hat g(x^i_t, m_t):= \hat \psi(m_t)(x^i_t),
\end{equation}
where $m_{t+1}=\hat f(m_t,\gamma_t)$ with  $m_1=P_X$.
\end{theorem}
\begin{proof}
The proof follows from the continuity and boundedness properties  described in Assumption~\ref{ass:Lipschitz_finite} and the strong law of large numbers. For more details, see~\cite[Theorem 4]{Jalal2019MFT}.
\end{proof}

For Model II, we make the following assumption.
\begin{assumption}\label{ass:LQ_bounded_n}
All matrices  in the agent dynamics~\eqref{eq:dynamics_LQ} and cost
function~\eqref{eq:cost_LQ}  as well as the covariance matrices  of the  initial states and  driving noises are uniformly bounded with respect to $n \in \mathbb{N}$.  In addition,  initial states and driving noises are independent random variables across agents at any time instant $t \in \mathbb{N}$.
\end{assumption}

For coupled dynamics, an additional  stability condition is required to ensure that the proposed  infinite-population strategy is stable under NS information structure. 
\begin{assumption}\label{ass:stability_hrwitz}
Matrix $\bar{\mathbf  A} + \bar{\mathbf B}\DIAG(\theta^\ast, \ldots,\theta^\ast)$ is Hurwitz.
\end{assumption}
It is to be noted that Assumption~\ref{ass:stability_hrwitz} holds  for decoupled dynamics, where $\bar A^j$ and $\bar B^j$ are zero, $\forall j \in \mathbb{N}_z$.  Define mean field $\mathbf m_{t}:=\VEC(m^1_t,\ldots,m^z_t)$ such that
\begin{equation}
\mathbf m_{t+1}=(\bar{\mathbf A}+ \bar{\mathbf B}\bar{\boldsymbol \theta}^\ast)  \mathbf m_t,
\end{equation} 
where  $ \mathbf m_1=\Exp{\bar{\mathbf x}_1}$.
\begin{theorem}\label{thm:LQ_NS_exact}
Let Assumptions~\ref{ass: stability_known_LQ},~\ref{ass:LQ_bounded_n} and~\ref{ass:stability_hrwitz} hold.
The following strategy is a solution  of Problem~\ref{problem2} with Model II
\begin{equation}
\hat u^{i}_t=\theta^\ast x^i_t+ \sum_{j=1}^z \alpha^{i,j} (\bar \theta^{j,\ast} \mathbf m_t - \theta^{\ast} m^j_t), \quad t \in \mathbb{N}.
\end{equation}
 Also, for the special case of weakly coupled  matrices, 
\begin{equation}
\hat u^{i}_t=\theta^\ast x^i_t+ \sum_{j=1}^z \alpha^{i,j} (\bar \theta^{j,\ast}  - \theta^{\ast})  m^j_t.  
\end{equation}
\end{theorem}
\begin{proof}
The proof follows from Assumptions~\ref{ass:LQ_bounded_n} and~\ref{ass:stability_hrwitz}, Theorem~\ref{thm:LQ_DSS},  the strong law of large numbers, and the fact  that  the optimal strategy and cost function are  bounded and continuous with respect to $n$. For more details, see~\cite[Theorem 3]{Jalal2019risk} and~\cite[Theorem 3.8]{arabneydi2016new}.    
\end{proof} 
\begin{remark}
\emph{When the system matrices are independent of the number of agents $n$,  the solution of the deep Riccati equation is also  independent of  $n$  for  the risk-neutral cost minimization~\cite[Chapter 3]{arabneydi2016new} and minmax optimization~\cite{Jalal2019LCSS}. However,  this is a rather special case,  and more generally,  such as in  risk-sensitive cost minimization problem~\cite{Jalal2019risk} and linear quadratic game~\cite{Jalal2019Automatica}, the solution  depends on $n$.}
\end{remark} 

\subsection{Numerical solutions}
To  numerically solve the Bellman equations in Theorems~\ref{thm:theorem1} and~\ref{thm:theorem 3}, one can  quantize   the space of probability measures on the local state and local action (i.e. $\mathcal{P} (\mathcal{X})$ and   $\mathcal{P}(\mathcal{U})$), and  use standard methods such as value iteration, policy iteration and linear programming, according to~\cite{Bertsekas2012book}. 

\begin{remark} 
\emph{In practice,   value iteration,  policy iteration and linear programming suffer  from the curse of dimensionality when the size of state and action spaces is large, unless some special structures are imposed such as a linear quadratic model. In Section~\ref{sec:RL}, we discuss more practical approaches that provide  more efficient solutions at the cost of losing the optimality.  It is worth mentioning that the convergence in policy space is often   faster than that in value space.}
\end{remark}

For Theorems~\ref{thm:LQ_DSS} and~\ref{thm:LQ_NS_exact}, one can use different techniques   to  solve the algebraic deep  Riccati equation~\eqref{eq:deep_Riccati}. Some common approaches include   invariant subspaces such as Schur method,  Newton-type iteration and  Krylov subspaces for large-scale problems. See~\cite{Bini2011} for more advanced methods.

\subsection{Decentralized implementation}
\subsubsection*{Model I with DSS}
Every agent $i \in \mathbb{N}_n$ can independently  solve the  dynamic program~\eqref{eq:bellman_exact_n}  upon the observation of  deep state. Since the resultant optimization problem is the same for all agents, the agents commonly choose the control  law $\psi^\ast$.  Then,    agent $i$ chooses its action $u^i_t$  based on  its local (private) state $x^i_t$ and deep state $d_t$  with respect to the probability distribution function $\psi^\ast(d_t)(x^i_t)=\gamma^\ast_t(x^i_t)$ at any time $t \in \mathbb{N}$ (see Theorem~\ref{thm:theorem1} for more details).
\subsubsection*{Model I with NS}
 For NS information structure, deep state is replaced by mean field, and   every agent solves the dynamic program~\eqref{eq:bellman_exact_inf}  by predicting the mean field $m_{t+1}$, given  current value $m_t$ and local law $\gamma_t$, according to~\eqref{eq:mean_field_dynamics}. The resultant performance converges  to that of the optimal one  as the number of agents goes to infinity; (see Theorem~\ref{thm:theorem 3} for more details). 

\subsubsection*{Model II with DSS}
Every agent solves one Riccati equation of the  same order as  an individual agent and one Riccati equation  whose order  is  $z$ times  greater than  that of an individual agent given in~\eqref{eq:deep_Riccati}. For the weakly coupled  case,  every agent solves $z+1$ Riccati equations of  an individual agent's order (Remark~\ref{remark:equivariant_Riccati}). Then, each agent computes its action $u^i_t$ based on its impact factors $(\alpha^{i,j})_{j=1}^z$,  local state~$x^i_t$ and deep state $\bar{\mathbf x}_t$ (see Theorem~\ref{thm:LQ_DSS} for more details).

\subsubsection*{Model II with NS}
 For NS information structure,  deep state is replaced by mean field and the resultant solution  asymptotically converges to the optimal solution as the number of agents goes to infinity.

\section{Main results for Problems~\ref{problem3} and~\ref{problem4}}\label{sec:RL}
 In the previous section, it was  assumed that the model is  completely known. In this section, we provide various techniques to approximate the proposed dynamic programs for the case when the model is not known completely.  

In general, there are two fundamental approaches  to learning  the solution of the proposed dynamic programs. The first one is called \emph{model-based} approach  which uses  supervised learning techniques to find the parameters of the models described  in Section~\ref{sec:formulation}, and then solve the planning problems presented in Section~\ref{sec:planning}. In short, this approach is an indirect method that obtains the  solution  by  constructing a model. The second approach, however, finds the solution directly without  identifying the model. This approach  is called \emph{reinforcement learning} (RL)  (also called \emph{model-free} or \emph{approximate dynamic program}). 

The advantage of the model-based approaches is that they are  more intuitional  because  they not only provide a solution but also construct a model.  However, for large-scale problems, it is more efficient to use  reinforcement learning methods as they directly search for the solution. In general, there are  two types of RL algorithms: off-line and on-line. In the former type, the exploration step is not a major  concern as  all states and actions can be  visited sufficiently often, given a rich set of data and/or simulator. In the latter type,  however, it is critical  to explore the model in such a way that   all  states and actions are visited sufficiently often.

%

\subsection{Model I}
In what follows, we  briefly present the main idea behind approximate value iteration, approximate policy iteration and approximate  linear programming. For the approximate value iteration, consider a one-step ahead update of Bellman equation~\eqref{eq:bellman_exact_n} as follows:
\begin{equation}\label{eq:approximate_VI}
\min_{\gamma }  \bar c(d,\gamma)+ \beta  \sum_{\tilde d} \Prob{\tilde d \mid d, \gamma} \tilde V( \tilde d)
\end{equation} 
where   $\gamma: \mathcal{X} \rightarrow \mathcal{P}(\mathcal{U})$ can be approximated  as:
\begin{equation}
\gamma(x)(u) \approx \frac{\theta(x,u)}{\sum_{u \in \mathcal{U}} \theta(x,u)}, \quad \theta(x,u) >0, \hspace{.1cm} x \in \mathcal{X},  \hspace{.1cm} u \in \mathcal{U},
\end{equation}
 such that  $\sum_{u \in \mathcal{U}}\gamma(x)(u)=1$. In particular, one can use a  normalized exponential distribution  with the following form:
\begin{equation}
\gamma(x)(u)  \approx \frac{e^{-\theta(x,u)}}{\sum_{u \in \mathcal{U}} e^{-\theta(x,u)}}.
\end{equation}
  It is also possible to approximate the value function linearly using feature-based architecture as follows:
\begin{equation}
V(\tilde d) \approx \tilde V(\tilde d)=\sum_{\ell=1}^L \tilde r_\ell \phi_\ell(\tilde d).
\end{equation}
Moreover, one can use a multi-step ahead update in~\eqref{eq:approximate_VI}. Similar function approximations can be used in  policy iteration and linear programming. For more details, see~\cite{Bertsekas2012book,Sutton2018introduction}. 

We present Algorithm~\ref{alg.RL-Qlearning}, a (model-free) Q-learning algorithm, wherein attention is restricted to  deterministic  strategies, i.e.,  $\gamma: \mathcal{X} \rightarrow \mathcal{U}$    for Model I.   Let $\mathcal{G}$ denote the set of mappings from the local state space  $\mathcal{X}$ to the  local action space $\mathcal{U}$. It is  also possible to  use  various function approximations to provide a more practical algorithm, albeit at the cost of  reduced  performance.

\alglanguage{pseudocode}
\begin{algorithm}[t!]
\small
\caption{Proposed Q-Learning Procedure}
\label{alg.RL-Qlearning}
\begin{algorithmic}[1] 
\State Set  $Q_1(d,\gamma)=0$ and $\eta_1(d,\gamma)=1, \forall d \in \Emp_n(\mathcal{X}),  \hspace{.1cm}\forall \gamma \in \mathcal{G}$.

\State At  iteration $k \in \mathbb{N},$ given any deep state $d \in \Emp_n(\mathcal{X}) $ and any local law $\gamma  \in \mathcal{G}$,  update   the  corresponding Q-function and  learning rate as follows:
\begin{align}
\begin{cases}
Q_{k+1}(d,\gamma)=(1-\eta_k(d, \gamma)) Q_k (d, \gamma)\\
\quad + \eta_k (d, \gamma) (\bar c'+\beta \min_{\gamma' \in \mathcal{G}} Q_k(d',\gamma')),\\
\eta_{k+1}(d,\gamma)=\lambda(k,\eta_k(d,\gamma)),
\end{cases}
\end{align}
where $\bar c'$ is the immediate cost, $d'$ is the next deep state, and $\lambda$  determines proper learning rates  $\eta_k \in [0,1], k \in \mathbb{N},$  such that
$\sum_{k=1}^\infty \eta_k(d,\gamma) = \infty$ and  $\sum_{k=1}^\infty (\eta_k(d,\gamma))^2 <\infty.$
  
  \State \noindent Set $k= k+1$, and go to step 2 until the algorithm terminates.
\Statex
\end{algorithmic}
  \vspace{-0.2cm}%
\end{algorithm}

\begin{theorem}\label{thm:RL}
Suppose that every pair of deep state and local law $(d,\gamma) \in \Emp_n(\mathcal{X}) \times \mathcal{G}$  is visited infinitely often  in Algorithm~\ref{alg.RL-Qlearning}. Then,  the following results hold:
\begin{itemize}
\item[ (a)]  For any $(d,\gamma) \in \Emp_n(\mathcal{X}) \times \mathcal{G}$, the Q-function $Q_k(d,\gamma)$  converges  to  $Q^\ast(d,\gamma)$ with probability one, as $k \rightarrow \infty$.
\item[ (b)] Let $g_k(\boldsymbol \cdot, d) \in \argmin_{\gamma \in \mathcal{G}}Q_k(d,\gamma)$ be a greedy strategy; then,   the performance of $g_k$ converges  to that of the optimal strategy $g^\ast$ given in Theorem~\ref{thm:theorem1}, when attention is restricted to deterministic  strategies.
\end{itemize}
\end{theorem}
\begin{proof}
The proof follows from the proof of  convergence of the Q-learning algorithm and Theorem~\ref{thm:theorem1}, which exploits  the fact that the Bellman  operator is   contractive with respect to the infinity norm. See \cite[Theorem 4]{Tsitsiklis1994asychronous}  for more details on the convergence proof of the Q-learning algorithm.
\end{proof}
Similar to Theorem~\ref{thm:RL}, one can use  a quantized space with quantization level $1/q$, $ q\in \mathbb{N}$, similar to the one proposed in~\cite[Theorem 6]{Jalal2019MFT},  to develop an approximate Q-learning algorithm  under NS information structure. The performance of the learned strategy converges to that of Theorem~\ref{thm:RL} as the number of agents $n$ and quantization parameter $q$ increase.
\begin{remark}
\emph{The above results can be extended to stochastic shortest path  problem where  $\beta=1$ under the condition that  there exists a special cost-free terminal state that absorbs all states  under any strategy~\cite{bertsekas1991analysis}.}
\end{remark}
\subsection{Model II}
For Model II, we use a model-free  policy-gradient method proposed in~\cite{fazel2018global}, and present Algorithm~\ref{alg.RL-LQ}.  Given a smoothing parameter $r>0$, let $P_r$ denote the  uniform probability distribution over the matrices of size $h_u \times h_x$, whose Frobenius norm is $r$. Similarly, let $\mathbf{P}_r$ denote the  uniform probability distribution over the matrices of size $zh_u \times zh_x$, whose Frobenius norm is $r$.


\alglanguage{pseudocode}
\begin{algorithm}[t!]
\small
\caption{Proposed Policy Gradient Procedure}
\label{alg.RL-LQ}
\begin{algorithmic}[1] 
\State  Initialize the number of agents $n$, number of trajectories $L$, control  horizon $T$, number of features $z$, feedback gains $(\theta_1, \bar{\boldsymbol \theta}_1)$, smoothing parameter $r$, and step size $\eta$. 

\State At  iteration $k \in \mathbb{N}$, run the following steps: 
\begin{itemize}
\item[] for $\ell=1:L$ 
\begin{itemize}
\item initialize states $\mathbf x_1=\VEC(x^1_1,\ldots,x^n_1)$;
\item given any agent $i \in \mathbb{N}_n$, use strategy~\eqref{eq:optimal_LO} with perturbed feedback gains: $\theta_k +\tilde \theta$ and $\bar{\boldsymbol \theta}_k + \tilde{\boldsymbol \theta}$, where $\tilde \theta \sim P_r$ and $\tilde {\boldsymbol \theta} \sim  \mathbf P_r$;
\item compute the cost  trajectories $c^\ell_{1:T}=\frac{1}{n} \sum_{i=1}^n c^\ell_{1:T}(i)$;
\end{itemize}
\item[] end
\item[] compute  $\nabla \bar  C \hspace{-.1cm}= \hspace{-.1cm}\frac{z^2h_xh_u}{TL r^2} \sum_{\ell=1}^L \sum_{t=1}^T \beta^{t-1} c^\ell_{t} \tilde{\boldsymbol \theta}$ and $\nabla \hat C=\frac{h_xh_u}{TL r^2} \sum_{\ell=1}^L \sum_{t=1}^T \beta^{t-1}  c^\ell_{t} \tilde \theta$;
\end{itemize}

  \State Update feedback gains: $\theta_{k+1}=\theta_k -\eta \nabla  \hat C_k$ and $\bar{\boldsymbol \theta}_{k+1}=\bar{\boldsymbol \theta}_k -{\eta} \nabla \bar C_k $.
  \State \noindent Let $k= k+1$, and go to step 2 until  the algorithm terminates.
\Statex
\end{algorithmic}
  \vspace{-0.2cm}%
\end{algorithm}

\begin{assumption}\label{ass:invertible}
The covariance matrices of initial states and driving noises are positive definite (i.e., they are invertible).
\end{assumption}
\begin{theorem}\label{thm:RL_LQ}
Let Assumptions~\ref{ass: stability_known_LQ},~\ref{ass:LQ_bounded_n} and~\ref{ass:invertible} hold. The performance of the  learned strategy $\{\theta_k,\bar{\boldsymbol \theta}_k\}$, given by Algorithm~\ref{alg.RL-LQ}, converges to that of the optimal strategy $\{\theta^\ast,\bar{\boldsymbol \theta}^\ast\}$ in Theorem~\ref{thm:LQ_DSS} with probability one, as $k \rightarrow \infty$.
\end{theorem}
\begin{proof}
The proof follows from~\cite[Theorem 9]{fazel2018global} and the decomposition proposed in~Theorem~\ref{thm:LQ_DSS}.
\end{proof}
Analogous to Theorem~\ref{thm:RL_LQ}, one can devise an approximate policy gradient algorithm under NS information structure, where  deep state is approximated by mean field. Note that Theorem~\ref{thm:RL_LQ} holds for arbitrary probability distribution (i.e., not necessarily  Gaussian and/or i.i.d. across agents).

\section{Numerical example}\label{sec:numerical}
The power grid is a complex   large-scale  network consisting of many decision makers   such as  users and service providers. Due to some fundamental  challenges such as  global warming,  limited fossil fuel and intermittent nature  of  renewable energy sources,  there is an inevitable need  for smart grid wherein the decision makers   intelligently interact with each other  and use  limited resources efficiently. As a result,  there has   been a  growing  interest  recently in network management of smart grid~\cite{fang2012smart}. In what follows, we provide a simple example showcasing the application of  our results in learning  the optimal resource allocation strategy.

\textbf{Example 1.} Consider a smart grid  with  $n \in \mathbb{N}$ users.  Let  $x^i_t \in \mathbb{R}$ denote the consumed  energy of  user $i \in \mathbb{N}_n$   at time $t \in \mathbb{N}$ and $\bar x_t$ denote the  weighted average of the total  energy consumption of  users, i.e., 
\begin{equation}
 \bar x_t= \frac{1}{n} \sum_{i=1}^n \alpha^i x^i_t,
\end{equation}
where $\alpha^i$ indicates the relative  importance (priority)  of  user  $i$ compared to others.  The linearized dynamics of  user $i$ is:
\begin{equation}
x^i_{t+1}=x^i_t + u^i_t + w^i_t,
\end{equation}
where $ w^i_t$  reflects the uncertainty in  energy consumption of user $i$ at time $t$.   The objective is to find a resource allocation strategy such that   the following cost is minimized:
\begin{equation}
\lim_{T \rightarrow \infty} \frac{1}{T} \Exp{\frac{1}{n}\sum_{t=1}^T   \sum_{i=1}^n (x^i_t)^\intercal Q x^i_t + (u^i_t)^\intercal R u^i_t +  \bar x_t^\intercal \bar Q \bar x_t },
\end{equation}
where the first two terms are the operational   cost of each user  and   the third term is  the  cost  associated with purchasing energy from a utility.

Suppose  that the information structure is deep state sharing, and let all users run Algorithm~\ref{alg.RL-LQ} as their energy management strategy.  Consider the following numerical parameters:
\begin{align}
&n=10, \quad A=1,  \quad B=1,\quad \bar Q=4, \quad R=1,
 \quad  Q=1,\\
 &\bar R=1, \quad r=0.15, \quad  \eta=0.3, \quad T=10, \quad L=100,\\
 &\beta=1,   \quad \alpha^{1:6}= \sqrt{0.5}, \quad \alpha^{7}=\sqrt{1.5}, 
  \quad \alpha^{8}= 1,\alpha^{9}=\sqrt{2},\\
  &  \alpha^{10}=\sqrt{2.5}, \quad w^i_t \sim \text{norm}(0,0.02), \quad x^i_1 \sim  \text{unif}(0,0.1),
\end{align}
where the  initial states and local noises are assumed to be i.i.d. random variables. 
\begin{figure}[t!]
\centering
\hspace{0cm}
\includegraphics[trim={0 8.4cm 0 8.5cm},clip, width=\linewidth]{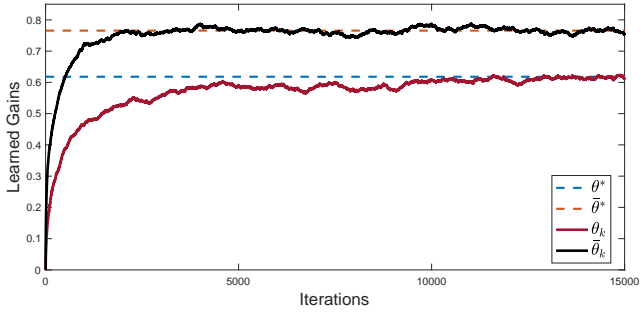}
\caption{The learning process in Example 1. }\label{figure}
\end{figure}
 The simulation results are provided in Figure~\ref{figure}, which gives the evolution of learning gains obtained by Algorithm~\ref{alg.RL-LQ} along with  the optimal gains. The figure shows that the learned strategy reaches a sufficiently small neighbourhood of the optimal one obtained by Theorem~\ref{thm:LQ_DSS},  after a few thousand iterations.

\section{Conclusions}\label{sec:conclusions}
In this paper,  the application of reinforcement learning  algorithms  in deep structured teams was studied for
 Markov chain and linear quadratic  models  with  discounted  and time-average cost functions. Two non-classical information structures were considered, namely, deep state sharing and no sharing.  Different  planning and reinforcement leaning algorithms  were proposed.  In particular,   it was shown that  the solution of a (model-free) Q-learning algorithm and  a (model-free) policy gradient algorithm converge  to the optimal solution of the Markov chain model  and linear quadratic model, respectively.   Finally, the obtained results were applied to smart grid in a simulation environment.

%

%
%


%

\bibliographystyle{IEEEtran}
\bibliography{Jalal_Ref}
\end{document}